\documentclass[11pt]{article}
\usepackage{lmodern}
\usepackage{amssymb,amsmath,amsthm}
\usepackage[T1]{fontenc}

\usepackage[margin=1in]{geometry}

\usepackage{algorithm}

\usepackage{paralist}

\usepackage{verbatim}

\newcommand{\ipcoskip}[1]{#1}
\newcommand{\ifipco}[2]{#2}
\newcommand{\fullv}[1]{#1}

\newcommand{\ipcoprooflater}[4]{\begin{proof}#4\end{proof}}
\newcommand{\ipcoproofoflater}[4]{\begin{proofof}{\Cref{#2}}#4\end{proofof}}

\newenvironment{proofof}[1]{\begin{proof}[Proof of #1]}{\end{proof}}

\newcommand{\ipcoqed}{}

\newcommand{\ssum}{\sum}

\usepackage{amsmath,amssymb}

\usepackage{verbatim}

\newenvironment{restatable}[2]{}{}

\usepackage[noend]{algpseudocode}
\usepackage{algorithm}

\ifipco{
\spnewtheorem*{inclaim}{Claim}{\bfseries\upshape}{\itshape}
\spnewtheorem{observation}{Observation}{\bfseries}{\upshape}
\spnewtheorem{mydefinition}[definition]{Definition}{\bfseries}{\upshape}
\renewenvironment{definition}{\begin{mydefinition}}{\end{mydefinition}}
\spnewtheorem{procedure}{Procedure}{\bfseries}{\upshape}
}{
\newtheorem{theorem}{Theorem}[section]
\newtheorem{lemma}[theorem]{Lemma}
\newtheorem*{inclaim}{Claim}

\newtheorem{corollary}[theorem]{Corollary}
\theoremstyle{definition}

\newtheorem{definition}[theorem]{Definition}
}

\usepackage[noend]{algpseudocode}
\usepackage{algorithm}

\usepackage{graphicx}
\usepackage{transparent}
\usepackage{xcolor}

\graphicspath{{pictures/}}

\definecolor{linkcol}{rgb}{0,0,0.35}
\definecolor{todocol}{rgb}{0.6,0,0}
\definecolor{citecol}{rgb}{0.1,0.35,0}

\usepackage{hyperref}

\hypersetup{
            pdftitle={Improved Approximation Algorithms for Inventory Problems},
            pdfborder={0 0 0},
            breaklinks=true}
\usepackage{cleveref}

\definecolor{mydarkgreen}{RGB}{0,100,0}

\newcommand{\E}{\mathbb{E}}

\DeclareMathOperator{\supp}{supp}
\DeclareMathOperator{\head}{head}
\DeclareMathOperator{\tail}{tail}

\newcommand{\disconn}{\divideontimes}

\newcommand{\fpsabv}{FPS}
\newcommand{\fps}{fractional path solution}

\newcommand{\sttf}{Steiner tree over time}

\newcommand{\saddc}{subadditive cover over time}

\newcommand{\subc}{submodular cover over time}
\newcommand{\irp}{IRP}
\newcommand{\sjrp}{SJRP}

\newcommand{\znn}{\mathbb{Z}_{\geq0}}

\newcommand{\Rnn}{\mathbb{R}_{\geq 0}}

\newcommand{\lovaszextension}{Lov\'asz extension}

\newcommand{\pprefix}[2]{#1_{\triangleright #2}}

\newcommand{\twl}{a}
\newcommand{\twr}{b}

\newcommand{\dW}{\mathcal{W}}

\newcommand{\relaxation}{\textsc{(lp)}}

\newcommand{\trunc}[2]{ {#1}^{| #2}}
\newcommand{\truncss}[3]{ {#1}^{#2| #3}}

\ifdefined\final
    \newcommand{\todo}[1]{\ignorespaces}

\newcommand{\nnote}[1]{\ignorespaces}

\else
\newcommand{\todo}[1]{{\em\color{red!50!black}TODO: #1}}

\newcommand{\nnote}[1]{{\em\color{blue!50!black}Neil: #1}}
\fi

\begin{document}
\title{Improved Approximation Algorithms for Inventory Problems\thanks{
Supported by Dutch Science Foundation (NWO) Vidi grant 016.Vidi.189.087.
    Part of this work was done while both authors were affiliated with the 
Department of Econometrics \& Operations Research, Vrije Universiteit Amsterdam.}
    } 
%
%
\author{Thomas Bosman\footnote{%
    Booking.com, Amsterdam, The Netherlands.
Email: \texttt{tbosman@gmail.com}.
}~~and
    Neil Olver\footnote{
    Department of Mathematics, London School of Economics and Political Science, London, United Kingdom.
Email:~\texttt{N.Olver@lse.ac.uk}.
    Visiting appointment at CWI, Amsterdam, The Netherlands.
}}
%
%
%
%
\maketitle              

\begin{abstract}
    We give new approximation algorithms for the submodular joint replenishment problem and the inventory routing problem, using an iterative rounding approach. 
    In both problems, we are given a set of $N$ items and a discrete time horizon of $T$ days in which given demands for the items must be satisfied.
    Ordering a set of items incurs a cost according to a set function, with properties depending on the problem under consideration.
Demand for an item at time $t$ can be satisfied by an order on any day prior to $t$, but a holding cost is charged for storing the items during the intermediate period;
the goal is to minimize the sum of the ordering and holding cost. 

    Our approximation factor for both problems is $O(\log \log \min(N,T))$;
    this improves exponentially on the previous best results.

\ifipco{\keywords{approximation algorithms \and iterative rounding \and inventory.}}{}

\end{abstract}

\section{Introduction}
The inventory problem studied in this paper captures a number of related models studied in the supply chain literature. One of the simplest is the dynamic economic lot size model~\cite{wagner1958dynamic}. 
Here we have varying demand for a single item over $T$ time units. Demand at time $t$ or later can be satisfied by an order at time $t$ 
(but not vice versa). 
For each day, there is a per-unit cost for holding the items in storage; 
there is also a fixed setup cost for ordering any positive quantity of the item, which is the same for each day. 
We want to decide on how many items to order on each day so as to minimize the total ordering and holding cost. 

The \emph{joint replenishment problem (JRP)} generalizes this problem to multiple items. We now have a unique holding cost for each day and each item, and a per item setup cost for ordering any quantity of that item. Furthermore, there is a general setup cost for ordering any items at all. This setup cost structure is called \emph{additive}. 
While having limited expressive power in comparison to some of the more complex generalizations that have been studied, the additive joint replenishment problem is long known to be NP-hard~\cite{arkin1989computational}. 
This problem has attracted considerable attention from the theory community in the past, resulting in a line of progressively stronger approximation algorithms~\cite{levi2006primal,levi2008constant}, the best of which gives an approximation ratio of 1.791~\cite{bienkowski2014better}.

A more general version of this problem uses an ordering cost structure in which the setup cost is a submodular function of the items ordered. This model, introduced by Cheung et al.~\cite{cheung2016submodular}, is called the \emph{submodular joint replenishment problem (SJRP)} and captures both the additive cost structure as well as other sensible models. 
In the same work, the authors give constant factor approximation algorithms for special cases of submodular cost functions, such as tree cost, laminar cost (i.e.,
coverage functions where the sets form a laminar family) and cardinality cost (where the cost is a concave function of the number of distinct items). For the general case, they provide an $O(\log NT)$ approximation algorithm. 

In the \emph{inventory routing problem (\irp)} setup costs are routing costs in a metric space. There is a root node and every item corresponds to a point in the metric. The setup cost for a given item set is then given by the length of the shortest tour visiting the depot and every item in the set.
The usual interpretation of the model is that the root node represents a central depot and every other point in the metric represents a warehouse to be supplied from the depot. (To streamline terminology with the joint replenishment problem, we will keep using the term ``item'' rather than ``location''.) 

The \irp{} has been extensively studied in the past three decades (see~\cite{coelho2013thirty} for a recent survey), but primarily from a computational perspective. 
But very little is known about its approximability.
Fukunaga et al.~\cite{fukunaga2014deliver} presented a constant factor approximation under the restriction that orders for a given item must be scheduled \emph{periodically}. 
This restriction appears to significantly simplify the
construction of an approximation algorithm, as prior to this work the best known polynomial time approximation algorithms gave (incomparable) $O(\log N)$~\cite{fukunaga2014deliver} 
\ipcoskip{(obtained by combining a constant factor approximation ratio for trees with standard metric embedding arguments)} 
and $O(\log T)$~\cite{nagarajan2016approximation} performance guarantees.

Nagarajan and Shi~\cite{nagarajan2016approximation} break the logarithmic barrier for both \irp{} and \sjrp, under the condition that
holding costs grow as a fixed degree monomial. 
This is a very natural restriction; in particular it captures the case where holding an item incurs a fixed rate per unit per day, 
depending only on the item. 
Building on their approach, we improve exponentially on their $O(\log T / \log \log T)$ approximation factor.
We also provide some general techniques to turn (sub)logarithmic
approximation algorithms in terms of $T$ into equivalent algorithms in terms of $N$;
and we are able to obtain results without restriction on the holding costs.
Our main contributions are summarized in the following theorems. 

\begin{theorem}
\begin{restatable}{thm}{irpapx}
  \label{thm:irpapx}
  There is a polynomial time $O(\log \log \min(N,T))$-approximation algorithm for the inventory routing problem. 
\end{restatable}
\end{theorem}

\begin{theorem}
\begin{restatable}{thm}{sjrpapx}
  \label{thm:sjrpapx}
  There is a polynomial time $O(\log \log\min(N,T))$-approximation algorithm for the submodular joint replenishment problem. 
\end{restatable}
\end{theorem}

We also mention the works on submodular partitioning problems~\cite{chekuri2011approximation,chekuri2011submodular,ene2013local}. 
In these problems, a ground set $V$ must be partitioned across $k$ different sets to minimize a submodular cost function. They use rounding of a relaxation based on the \lovaszextension{} to unify and improve several prior results. Their approach inspired our use of the \lovaszextension{} in the rounding algorithm for \sjrp.

\section{Preliminaries, model and technical overview}\label{sec:preliminaries}

We use $\log$ for the base $2$ logarithm. 
We write $[k]$ for $\{1,\dots,k\}$, and $[k,\ell]$ for $\{k, k+1, \ldots, \ell \}$, for any integers $k \leq \ell$.
\ipcoskip{The support $\supp(w)$ of a function $w : X \to \mathbb{R}$ is the subset of $X$ for which $w$ maps to nonzero values. }

The general framework of the inventory problems we investigate is defined by a set of items 
$V$ of size $N$, 
ordering cost function $f: 2^{V}\to \Rnn$ and a time horizon $[T]  =\{ 1,\dots,T\}$.  
We will assume throughout this paper that $f$ is monotone and subadditive, with $f(\emptyset) = 0$.
We will typically refer to the atomic time units as \emph{days}. %

For each item $v \in V$ and day $t$, there is a demand $d_{vt}\geq 0$. 
The collection of item-day pairs for which there is positive demand is denoted $D := \{(v,t): d_{vt} >0\}$. 
Demand for day $t$ can be satisfied on or before day $t$.
If we satisfy demand for item $i$ on day $t$ using an order on some day $s<t$, we need to store the items in the intervening days, and we pay a holding cost of $h_{st}^v$ per unit we store. %
The magnitude of the demand only plays a role in the holding cost; the ordering cost is determined by the unique items ordered and is independent of how many units are ordered.

Given these inputs, we need to place an order for items to be delivered on each day so as to minimize the total ordering cost plus the holding cost. Since the cost of delivering an item does not depend on the size of the order and we want to store items as briefly as possible, it is always optimal to deliver just enough units of an item to satisfy demand until the next order for that item is scheduled. 
Hence, once we decide which items to order on which days, the optimal schedule is completely determined. 

The inventory routing problem is the special case where we have a metric on $V$, some distinguished root node $r\in V$, and the ordering cost $f(S)$ of a set $S\subseteq V$ is the minimum possible length of 
a tree containing $S \cup \{r\}$.
Here we differ from the usual definition, where $f(S)$ is defined to be the length of a shortest tour on $S \cup \{r\}$;
but as is well known, these two definitions differ only by a factor of at most 2, which will not concern us.
The submodular joint replenishment problem is the special case where $f$ is submodular (in addition to the required properties already listed).

An integer programming formulation for this problem is given in~\eqref{ilp:subadditive}. 
Here the variable $y_t^S$ indicates whether item set $S$ is ordered on day $t$, and $x_{st}^i$ indicates whether the demand for item $v$ on day $t$ is satisfied by an order on day $s$. 
\begin{equation}
  \label{ilp:subadditive}
  \begin{aligned}
    \text{minimize }\quad & \ssum_{t\in [T]} \ssum_{S \subseteq V} f(S) y_t^S + \ssum_{t\in [T]} \sum_{v\in V}  \ssum_{s \leq t} d_{vt}h^v_{st} x_{st}^v  \\
    \text{subject to } \quad & x_{st}^v \leq \ssum_{S : v \in S} y_s^S \qquad \forall (v,t) \in D, s\leq t \\
    & \ssum_{s \leq t} x_{st}^v = 1 \qquad \forall (v,t) \in D  \\
    & y_t^S, x_{st}^v \in \{0,1\}\qquad  \forall v \in V,s \leq t, S \subseteq V
  \end{aligned}
\end{equation}
Let \relaxation{} denote the LP relaxation obtained by replacing the integrality constraints of ILP~\eqref{ilp:subadditive} by nonnegativity constraints; this LP has an exponential number of variables. 
To efficiently solve \relaxation{}, it suffices to provide an efficient separation oracle for the dual.
This can be done for both \sjrp{} and \irp{} (in the latter case, in an approximate sense);
see~\cite{nagarajan2016approximation}.

Nagarajan and Shi~\cite{nagarajan2016approximation} show that in order to round \relaxation{}, it suffices to round an associated covering problem.
Essentially, given a solution $(\hat{x},\hat{y})$ to \relaxation{}, we require that each demand $(v,t) \in D$ is served within an interval $[s'(v,t), t]$, where $s'(v,t)$ is the median of the distribution $(\hat{x}^v_{st})_{s \leq t}$.
Serving $(v,t)$ anywhere within this interval will incur cost at most twice what the fractional solution pays;
and moreover, they show that enforcing this restriction cannot make the optimal solution more than a constant factor more expensive.
The holding costs can then be dropped from the objective function.
All in all, we obtain an instance of the following \emph{\saddc} problem:
for each item $v \in V$ we are given an associated set of \emph{demand windows} $\dW_v \subseteq \{ [s,t] : s \leq t \in [T]\}$.
We must choose a subset $S_t \subseteq V$ for each day $t\in [T]$ such that every item $v \in V$ is covered in each of its demand windows---that is, $v \in S_r$ for some $r \in [s,t]$, for each $[s,t] \in \dW_v$.
The goal is to find a feasible solution minimizing the total cost $\sum_{t\in [T]} f(S_t)$. 

We also associate the canonical LP~\eqref{lp:saddc} with the \saddc{} problem. 
\newcommand{\lpeqssaddc}{
  \begin{aligned}
    \text{minimize }\quad & \ssum_{t\in [T]} \ssum_{S \subseteq V} f(S) y_t^S  \\
    \text{subject to } \quad & \ssum_{r \in [s,t]} \ssum_{S : v \in S} y_r^S \geq 1 \qquad \forall [s,t] \in \dW_v, \quad \forall v\in V  \\
    & y_t^S\geq 0 \qquad  \forall  t \in [T], S \subseteq V
  \end{aligned}
}
\begin{equation}
  \label{lp:saddc}
\lpeqssaddc{}
\end{equation}
Our goal, given a fractional solution to this LP, is to round it to an integral solution.
Note that the instance of \saddc{} is constructed from a solution $(\hat{x}, \hat{y})$ to \relaxation{} in such a way that $\hat{y}$ is already a feasible fractional solution to \eqref{lp:saddc}.

We now come to our first contribution.
We show that this reduction can be taken much further: 
we can reduce to covering problems where the set of intervals have a very special structure.
This structure, which we call \emph{left aligned}, shares many of the benefits of a laminar family.
For example, they have a natural notion of depth, which is always logarithmically bounded by $T$;
the approximation factors of our algorithms are essentially logarithmic in this depth.
We describe this reduction, which is rather general and applies identically to both \sjrp{} and \irp, in \Cref{sec:reduction}.
We also show, again generally, how the time horizon $T$ can be polynomially bounded in terms of the number of items $N$.

So to obtain our main theorems, it suffices to give $O(\log \log T)$-factor approximation algorithms for these well-structured covering variants of \sjrp{} and \irp.
Here the approaches diverge; we give rather different algorithms for these two problems, albeit based on iterative rounding~\cite{jain2001factor} in both cases.

For \irp, the algorithm uses randomized iterative rounding~\cite{byrka_2013_steiner} of a certain path-based relaxation. 
We can show that after sampling every path in the support of the relaxation $O(\log \log T)$ times, we can remove a constant fraction of the edges and reorder the remaining paths such that we retain a feasible solution.  
Details are given in \Cref{sec:irp}.

For \sjrp, by contrast, the iterative rounding approach is naturally deterministic in nature. 
Instead of randomly rounding item sets in the support of a relaxation, we carefully try to pick a set for each day such that we win a constant fraction of its cost back in the subsequent reduction of the cost of the relaxation. 
If such a set cannot be found, we show that we can shrink the time horizon $T$ by merging some adjacent time units
(or put differently, we are able to remove the bottom ``leaf'' layer of the left-aligned family); we can then recurse on the smaller instance.
We discuss this further in \Cref{sec:sjrp}.

\section{Reducing to structured covering problems}
\label{sec:reduction}
The results of this section will not use any properties of the ordering cost function $f$ that differ between IRP and SJRP.
All that we will need, other than the general properties of $f$ assumed at the start of \Cref{sec:preliminaries}, 
is that we are given an (approximately) optimal solution to the LP relaxation of~\eqref{ilp:subadditive}.

  Let $\mathcal{D} = \{[k2^i+1, (k+1)2^i] : i, k \in \znn \}$ denote the family of dyadic intervals over the nonnegative integers;
  the value of $i$ for one of these intervals in $\mathcal{D}$ we call the \emph{level} of that interval.
\begin{definition}
  \label{def:dyadaligned}

  A family of intervals $\mathcal{F}  \subseteq \{ [s,t]: s,t \in \mathbb{Z}_{\geq 0}\}$ is called 
  \begin{itemize}
    \item \emph{left aligned} if for all $[s,t]\in \mathcal{F}$ there exists $[s,t'] \in \mathcal{D}$ with $t'\geq t$, 
    \item \emph{right aligned} if for all $[s,t] \in \mathcal{F}$ there exists $[s',t]\in \mathcal{D}$ with $s'\leq s$. 
    \end{itemize}
The \emph{level} of an interval $[s,t] \in \mathcal{F}$ is the level of the minimal interval of $\mathcal{D}$ containing $[s,t]$.
\end{definition}
\ipcoskip{Observe that if $\mathcal{F}$ is both left and right aligned, then it is a laminar family.}
\fullv{
An example of a left-aligned family is shown in Figure~\ref{fig:aligned}.

\begin{figure}[htb]
  \centering
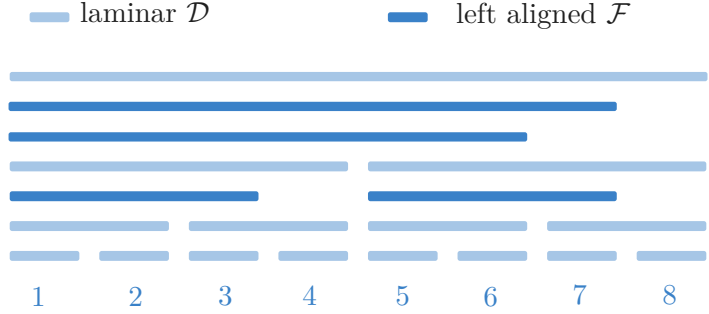
\caption[Example of laminar and left-aligned families.]{Example of laminar and left-aligned families of intervals on $\{1,\dots,8\}$. 
  \label{fig:aligned}
}
\end{figure}
}

We will call an instance of \saddc{} \emph{left (right) aligned} if 
$\bigcup_{v \in V} \dW_v$ is left (right) aligned.

\begin{theorem}
  \label{thm:leftaligned}
  At the loss of a constant factor, we can reduce an instance of the \saddc{} problem to a pair of instances, one left aligned and the other right aligned. 
  \ipcoskip{\footnote{See Figure~\ref{fig:alignedreduction} for an illustration.}}
\end{theorem}
\ipcoprooflater{reduction}{thm:leftaligned}{The proof is given in the appendix.}
    {
      Let $y$ be a solution to \eqref{lp:saddc}. 
      We will first generate two new instances of the \saddc{} problem, one being left aligned and the other right aligned. 

      Given an interval $[s,t]$, define the \emph{right-aligned part} $R([s,t])$ and the \emph{left-aligned part} $L([s,t])$ by
  \[ R([s,t]) = [s, k2^i] \quad \text{ and  } \quad L([s,t]) = [k2^i+1,t] ,\]
  where $i,k$ are integers such that $k2^i\in[s,t]$ and $i$ is maximal. If $k2^i = t$, then $L([s,t]) = \emptyset$, and if $k2^i + 1=s$ then $R([s,t]) =\emptyset$ by convention.
  It is clear from this definition that $\{ L([s,t]): v \in V, [s,t] \in \dW_v\}$ forms a left-aligned family, and similarly the right-aligned parts form a right-aligned family. 

  Any LP solution must cover every item by at least half in either the right-aligned or left-aligned part of its demand window. 
      For each $v\in V$ and demand window $[s,t]\in \dW_v$, if $L([s,t])$ receives half a unit of coverage under $y$, 
      add $L([s,t])$ as a time window for $v$ in the left-aligned instance; otherwise put $R([s,t])$ in the right-aligned instance.

      It is immediate from the way in which we constructed the two instances that $2y$ is a feasible solution to each.
    Hence the combined cost of the optimal solutions to the LP relaxations of the generated instances is at most 4 times that of the original instance. 
    Furthermore, we can translate integral solutions to the left and right
    aligned instances back to one for the original instance by adding them together, which does not increase the cost by subadditivity of $f$. 
}
Right-aligned instances can be handled identically to left-aligned ones, so we consider only left-aligned instances in the sequel.

\ipcoskip
{
\begin{figure}[htb]
  \centering
\def\svgwidth{0.95\textwidth}
\begingroup%
  \makeatletter%
  \providecommand\color[2][]{%
    \errmessage{(Inkscape) Color is used for the text in Inkscape, but the package 'color.sty' is not loaded}%
    \renewcommand\color[2][]{}%
  }%
  \providecommand\transparent[1]{%
    \errmessage{(Inkscape) Transparency is used (non-zero) for the text in Inkscape, but the package 'transparent.sty' is not loaded}%
    \renewcommand\transparent[1]{}%
  }%
  \providecommand\rotatebox[2]{#2}%
  \newcommand*\fsize{\dimexpr\f@size pt\relax}%
  \newcommand*\lineheight[1]{\fontsize{\fsize}{#1\fsize}\selectfont}%
  \ifx\svgwidth\undefined%
    \setlength{\unitlength}{474.76780303bp}%
    \ifx\svgscale\undefined%
      \relax%
    \else%
      \setlength{\unitlength}{\unitlength * \real{\svgscale}}%
    \fi%
  \else%
    \setlength{\unitlength}{\svgwidth}%
  \fi%
  \global\let\svgwidth\undefined%
  \global\let\svgscale\undefined%
  \makeatother%
  \begin{picture}(1,0.32230477)%
    \lineheight{1}%
    \setlength\tabcolsep{0pt}%
    \put(0.04389825,0.00826472){\color[rgb]{0,0,0}\makebox(0,0)[lt]{\lineheight{1.25}\smash{\begin{tabular}[t]{l}original time windows $\mathcal{W}_v$\end{tabular}}}}%
    \put(0.74713577,0.00844492){\color[rgb]{0,0,0}\makebox(0,0)[t]{\lineheight{1.25}\smash{\begin{tabular}[t]{c}split into left/right aligned families\end{tabular}}}}%
    \put(0,0){\includegraphics[width=\unitlength,page=1]{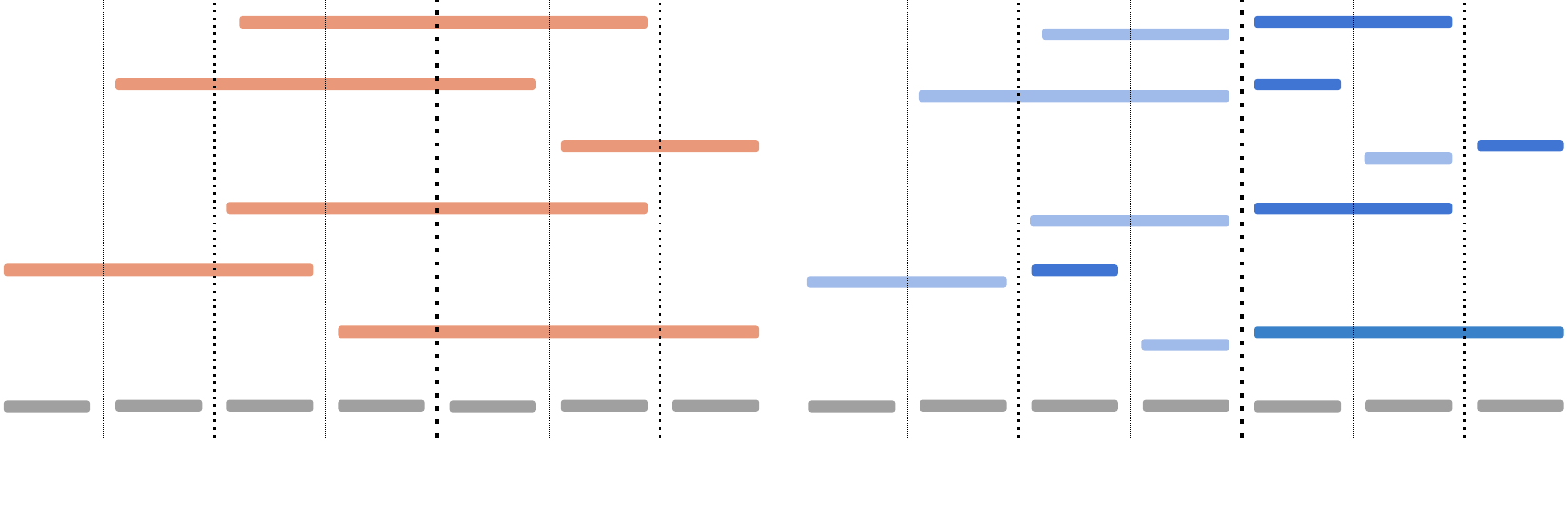}}%
  \end{picture}%
\endgroup%

\caption[Reducing instance to left and right aligned families.]{On the left an arbitrary collections of intervals is shown. On the right the intervals are split, resulting in a left-aligned (light blue) and right-aligned (dark blue) family of intervals.}
  \label{fig:alignedreduction}
\end{figure}
}

\ifipco{
\subsection{Bounding the time horizon, and further simplifications}\label{sec:timebound}
}{
\subsection{Bounding the time horizon}\label{sec:timebound}
}
\ifipco{Due to space constraints, we defer the proof of the following to the full version.}
    {Here we will show how to reduce the problem to the case where the time horizon is bounded by the size of the item set $N$. This allows us to focus on proving approximation ratios in terms of $T$, as they carry over to approximation ratios in terms of $N$ without further adjustments.}
 \begin{theorem}
   \label{thm:timebound}
   At the loss of a constant factor we can reduce a left-aligned \saddc{} problem to a polynomial-sized collection of left-aligned \saddc{} problems with time horizons equal to $N^2$. 
 \end{theorem}
 This theorem could already be used to improve the dependence of the Nagarajan-Shi algorithm from $O\left(\tfrac{\log T}{\log \log T}\right)$ to $O\left(\tfrac{\log \min\{T,N\}}{\log\log \min\{T,N\}}\right)$.
\ifipco{
 It also allows us to make the simplifying assumption that each item has positive demand on exactly one day, by 
 making a copy of an item for each day in which it has a positive demand (with $T \leq N^2$, this only increases the number of items by a polynomial factor). Finally, we also assume that 
 $T = 2^{2^k}$ for some $k \in \mathbb{N}$; if not, simply round up.
 Call an instance with all these properties (including being left aligned) \emph{nice}; 
 we will assume throughout the remainder of the paper that we are working with a nice instance.
}

\fullv{
    The remainder of this subsection is devoted to proving this theorem.
We will sometimes write $f(v)$ as shorthand for $f(\{v\})$ in what follows.

First, we show that the cost of singleton items are not too far apart.
\begin{lemma}
  \label{lem:wellseparated}
  At the loss of a constant factor, we can reduce the left-aligned \saddc{} problem to a collection of left-aligned \saddc{} problems with $\min_{v\in V} f(v)\geq \frac1{N} \max_{v\in V} f(v)$. 
\end{lemma}
\begin{proof}
    We will iteratively split the item set $V$ into a sequence of sets $V_1, V_2,\dots$ with items of descending singleton cost $f(v)$. 
    Initially, let $V_1 = V$. Then, for $j=1,\ldots, N$ and while $|V_j|>0$, let $\mu_j =\max_{v\in V_{j}} f(v)$. Remove all items $v$ with $f(v) < \mu_j/|V_j|$ from $V_j$ and add them to $V_{j+1}$.

  Now each set $V_j$ induces an instance of the required form. We claim that if we treat all these instances separately and take the union of their schedules the resulting schedules cost at most three times the original optimum.  
  The following construction gives a set of schedules that satisfies this bound. 

  Take an optimal solution $S_1,\dots,S_T$ to the original instance and generate a schedule for each instance $V_j$ by selecting item set $S_t^j := S_t \cap V_j$ on day $t \in [T]$. 

  On each day $t$, the total cost of the resulting schedule is dominated by the cost of two item sets $S_t^j, S_t^{j+1}$ where $j$ is the smallest index for which $S_t^j$ is nonempty. To be precise, it holds that:  
  \[    \sum_{k \geq j+2} f( S_t^k)\leq
      \sum_{v\in \bigcup_{k \geq j+2 }S_t^{k}} f(v) \leq 
\sum_{v\in \bigcup_{k \geq j+2 }S_{t}^{k}}  \mu_{j+1} / \Bigl|\bigcup_{k \geq j+1} V_k\Bigr| \leq
\mu_{j+1}
.\]

Here, in the first inequality we use subadditivity of $f$. For the second, we use the definition of the cutoff value for items that are removed from $V_{j+1}$ and put in $V_{j+2}$. In the last inequality, we use that $\bigcup_{k \geq j+2 }S_{t}^{k} \subseteq \bigcup_{k \geq j+1} V_k$ and therefore $|\bigcup_{k \geq j+2 }S_{t}^{k} |\leq |\bigcup_{k \geq j+1} V_k|$.  

Then, using $\mu_j = \max_{v \in V_j} f(v)$, $\max_{v\in V_{j+1}} f(v) \leq \min_{v\in V_j} f(v)$, and $S_j^t \subseteq V_j$, we get that:
\[\sum_{k \geq j+2} f( S_t^k)\leq \mu_{j+1} \leq  \min_{v\in V_j} f_v \leq f(S_j^t) . \]
In particular, the cost of the union of the schedules is at most  
\[f(\bigcup_{k \geq j }S_t^{k}) \leq
\sum_{k\geq j} f(S_t^k) \leq f(S_t^j) + f(S_t^{j+1}) + f(S_t^j) \leq 3f(S_t), \]
as required. 
Finally, we note we did not touch the demand windows and therefore the instances generated are also left aligned. This finishes the proof.
\end{proof}

We also need the following sparsity result on near-optimal solutions to LP relaxation~\eqref{lp:saddc} of the \saddc{} problem. It allows us to assume that every day, the fractional coverage of the sets in the LP solution adds up to at least a full unit, or else it covers no items at all. The intuition for this is that if it is not the case, we can group time intervals with less than a unit coverage, and copy all coverage inside those intervals to the endpoints.

 \begin{lemma}
   \label{lem:sparsity}
   Take a solution $y$ to \eqref{lp:saddc}. Then there exists a solution $\bar{y}$ whose cost is within twice the cost of $y$, 
   so that for each day $t \in [T]$, 
   \[  \sum_{S\subseteq V} \bar{y}_t^S \geq 1 \quad \text{ or  } \quad \sum_{S\subseteq V} \bar{y}_t^S  = 0.\] 
 \end{lemma}
 \begin{proof}
   Copy $y$ to a new solution $\bar{y}$. Call a day $t$ \emph{bad} if $\sum_{S\subseteq V} \bar{y}_t^S \in (0,1)$ and \emph{good} otherwise. We iteratively fix $\bar{y}$, maintaining that at each point there exists an index $k$ such that the first 
   $k$ days are good and $\sum_{t=1}^k \sum_S f(S) \bar{y}_t^S \leq \sum_{t=1}^k \sum_S f(S)  2{y}_t^S$.   
  
   Starting with $k =0$, suppose the first $k$ days are fixed. Find the first day $\ell>k$ that is bad, and let $m>\ell$ be the first index such that the fractional coverage of sets on the interval $[\ell, m]$ sums to at least one, i.e., $\sum_{t=\ell}^m \sum_{S\subseteq V} y_t^S \geq 1$, if such an index exists. We will deal with the other case, i.e., $\sum_{t=\ell}^T \sum_{S\subseteq V} y_t^S < 1$ later. 

   Since $m$ is the minimal, it must be that ${y}$ fractionally covers strictly less than a full set on the interval $[\ell, m-1]$. It follows that every node $v$ whose time window overlaps with $[\ell,m]$ must have $\ell \in F_v$ or $m\in F_v$. So, we can copy all coverage inside $[\ell,m]$ in $y$ to each of the endpoints $\ell$ and $m$ and retain a feasible solution, i.e.,
   \[ \bar{y}_\ell \gets  \sum_{t=\ell}^m y_t \quad \text{ and } \quad  \bar{y}_m \gets \sum_{t=\ell}^m y_t \quad \text{ and } \quad \bar{y}_t \gets 0 \quad \text{for} \quad  \ell< t< m. \]

   After doing this all days up to $m$ are good, and the additional cost can be charged against $y_\ell,\dots,y_m$, so we maintain the cost bound. We only need to deal with the edge case where $\sum_{t=\ell}^T \sum_{S\subseteq V} y_t^S< 1$. In this case, any node $v$ whose time window overlaps with $[\ell, T]$, must have $\ell-1 \in F_v$, otherwise it could not receive at least a unit coverage under $y$. Hence, we can copy all coverage to $\ell-1$ without losing feasibility:
   \[ \bar{y}_{\ell-1} \gets  \sum_{t=\ell}^T y_t \quad \text{ and }\quad \bar{y}_t \gets 0 \quad \text{for} \quad  \ell< t\leq  T.
   \]
   The cost of this action can be charged against $y_{\ell},\dots,y_T$. Since, we already know $\ell-1$ is a good day and we increase its coverage, we only need to worry about days $\ell$ to $T$, which get set to $0$ and thus become good. So, we conclude all days are good and $\bar{y}$ costs at most twice the cost of $y$, finishing the proof. 
   \end{proof}

   \begin{proofof}{\Cref{thm:timebound}}
     By Lemma~\ref{lem:wellseparated} we can reduce a left-aligned instance to a collection of instances with $\min_{v \in V} f(v) \geq \frac1{N} \max_{v} f(v)$.
     By Lemma~\ref{lem:sparsity}, we may assume we have for each of those instances an LP solution with $\sum_{S \subseteq V} y_t^S \geq 1$ or $\sum_{S \subseteq V} y_t^S =0$ for all $t$. At this point, we might as well delete the days $t$ with $\sum_{S \subseteq V} y_t^S=0$ from the instance. So, we can assume that $\sum_{S \subseteq V}y_t^S \geq  1$ for all $t$.  

Suppose that we simply schedule the entire set $V$ at time $t=N^2$. By Lemma~\ref{lem:wellseparated}, this costs at most 
\[f(V) \leq N \max_{v \in V} f(v) \leq N^2 \min_{v \in V} f(v).\]  
   By Lemma~\ref{lem:sparsity}, on the other hand, the cost of the LP solution for any day $t$ is at least 
   \[\sum_{S \subseteq V} f(S) y_t^S \geq \sum_{S \subseteq V} (\min_{v \in V} f(v) ) y_t^S \geq   \min_{v \in V} f(v) .\] 
   So, the LP solution for the first $N^2$ days costs at least $N^2 \min_{v \in V} f(v)$. 
   This means that we can charge the cost $f(V)$ of covering all elements at day $N^2$ against the LP, losing only a constant factor in the cost of the solution. 
  
   If we cover all elements on day $N^2$, we effectively reset the instance and we only need a schedule for the first $N^2$ days. Hence, we can delete any remaining time periods for the instance and we complete the proof. 
 \end{proofof}
 }

\ipcoskip
{
\subsection{Further simplifications}
\label{sec:reductionsummary}
Finally, we make some assumptions that are mainly useful to reduce notation in the coming sections. 
We assume that every item $v\in V$ has only one time $t$ for which $d_{vt} > 0$. 
If not we make multiple copies of each item, one for every day $t$ with positive demand.
\Cref{thm:timebound} ensures that this only increases the number of items by at most a factor of $N^2$.
We also assume that $T$ is of the form $2^{2^k}$ for $k\in \mathbb{N}$, so that $\log \log T$ is a positive integer. 
If this is not the case, we can simply round up $T$.
Our assumptions are summarized by the following definition.
\begin{definition}
  \label{def:niceinstance}
  An instance of \saddc{} is \emph{nice} if

  \smallskip

  \begin{compactenum}[(i)]
    \item the time windows form a left-aligned family,
    \item each item has positive demand on exactly one day, and
    \item $T = 2^{2^k}$ for some $k\in \mathbb{N}$.
  \end{compactenum}
\end{definition}
}

\section{Steiner tree over time}\label{sec:irp}
So in order to prove \Cref{thm:irpapx}, we need to give an $O(\log \log T)$-approximation algorithm to nice instances of \saddc{}, for the appropriate class of order functions $f$.
More precisely, $V$ is the set of nodes of a semimetric space with distance function $c : V\times V  \to \Rnn$; a root node $r \in V$ is specified, and for all other nodes $v \in V \setminus \{r\}$, a time window $F_v  = [\twl_v, \twr_v] \subseteq [T]$ is given.
Since $f(S)$ denotes the cost of a cheapest tree containing $S \cup \{r\}$, we will consider a solution to be described by a collection of trees $(\mathcal{T}_t)_{t \in T}$, all containing the root.
To be feasible, each node $v \neq r$ must be contained in $\mathcal{T}_t$ for some $t \in F_v$.
The cost of a tree $\mathcal{T}$ (i.e., the sum of the length of its edges) is denoted $c(\mathcal{T})$;
the objective is to minimize the total cost $\sum_t c(\mathcal{T})$. %
We will call this the \emph{Steiner tree over time problem}.%
\fullv{\footnote{
    The \sttf{} problem should not be confused with the problem known in the literature as the \emph{minimum spanning tree with time windows} problem~\cite{solomon1986minimum}.   
    This also involves a root and a set of nodes with time windows, but the goal is to construct a single tree, and to provide an order on the nodes such that a vehicle travelling at unit speed on the tree (optionally waiting at nodes) can visit the nodes within their time windows.
}
}

For ease of notation, we associate with each node $v$ a \emph{level} $\ell(v)$ in the natural way, namely as the level of $F_v$ in the left aligned family $\bigcup_{v \in V\setminus \{r\}} F_v$.

\ipcoskip{\subsection{Fractional path relaxation}}
The main part of our result works by iteratively massaging a specific type of fractional solution until it becomes integral. We now describe this type of solution.

We let $\mathcal{P}$ denote the collection of directed paths in $V$. For each such directed path $P \in \mathcal{P}$, let $P \odot_t v$ signify that $P$ connects $v$ to $\mathcal{T}_t$, i.e., $P \odot_t v$ if there is a directed subpath on $P$ from $v$ to a node in the tree $\mathcal{T}_t$ containing the root on day $t$. 
If $v\in \mathcal{T}_t$ we let $P\odot_t v$ hold for all $P$ by convention.

\begin{definition} %
  \label{def:fpsol}
  A \emph{fractional path solution} (\fpsabv{}) is a pair $(\mathcal{T}, w)$, giving for each day $t$ a tree $\mathcal{T}_t$ rooted at $r$ and weights $w_t : \mathcal{P} \to [0,1]$, with the property that
  \ifipco{%
      \[ \ssum_{t\in F_v} \ssum_{P \in \mathcal{P}:  P\odot_t v} w_t(P) \geq 1 \qquad \forall v\in V \setminus \{r\} .\] 
      }{%
      \[ \sum_{t\in F_v} \sum_{\substack{P \in \mathcal{P}: \\ P\odot_t v}} w_t(P) \geq 1 \qquad \forall v\in V \setminus \{r\} .\] 
  }
  The cost of the \fps{} is given by the sum of the cost of the trees and the cost of the paths weighted by $w$: 
  \[\ssum_{t \in [T]} \ssum_{P \in \mathcal{P}} w_t(P) c(P) + \ssum_{t \in [T]} c(\mathcal{T}_t) .\]
\end{definition}

Note that a \fps{} with $w_t(P) \in  \{0,1\}$ for all $P$ and $t$ yields a feasible solution to the \sttf{} problem. 
Moreover, we can start with a solution $y$ to \eqref{lp:saddc} and convert it to a \fps{} at the loss of a constant
(or more directly, we can solve a compact LP corresponding to \fps{}). %
To do this, start by initializing all trees $\mathcal{T}_t$ to contain only the root. 
    Then for each day $t$ and set $S$ in $\supp(y_t)$, construct a minimum spanning tree on $S \cup \{r\}$, 
    and use this to define a path $P$ to $r$ that contains $S$ and has cost at most twice the cost of this spanning tree (simply shortcut the doubled tree). Add $P$ to the solution with weight $w_t(P) = y_t^S$.

Hence, we focus on turning a \fps{} into an integral one without losing too much in the cost of the solution. 

We need some preliminary notation and definitions.
Given a directed path $P$, we use $V(P)$ and $E(P)$ to denote the node set and edge set, respectively.
We similarly define $V(\mathcal{T})$ and $E(\mathcal{T})$ for a tree $\mathcal{T}$.
\ifipco{The head and tail of $P$ are denoted $\head(P)$ and $\tail(P)$ respectively.}%
{The head of $P$, denoted $\head(P)$, is the unique node in $P$ without an outgoing edge; similarly $\tail(P)$ is the unique node of $P$ without an incoming edge.}
Given a collection of edges $R$, we denote by $P\disconn{R}$ the collection of directed paths obtained by deleting all edges in $R$ from $P$, disconnecting $P$ into multiple subpaths (we allow paths consisting of only a single node). 
Given an edge $e\in E(P)$, we use $\pprefix{P}{e}$ to denote the path in
$P\disconn{} \{e\}$ containing $\tail(P)$. 
Given a tree $\mathcal{T}$ and path $P$ with $\head(P)$ in $\mathcal{T}$, \emph{adding} $P$ to $\mathcal{T}$ (which we may write as $\mathcal{T} + P$) results in a spanning tree of the union of $\mathcal{T}$ and $P$.
In particular, $\mathcal{T} + P$ is a tree, costing no more than $c(\mathcal{T}) + c(P)$, and spanning $V(\mathcal{T}) \cup V(P)$.

The rounding algorithm will consist of a number of iterations.
Each iteration will increase the size of the integral part $(\mathcal{T}_t)_{t \in [T]}$, while reducing the size of the fractional part $(w_t)_{t \in [T]}$, until the solution is entirely integral.
We will ensure that the cost increase in the integral part is an $O(\log \log T)$ factor times the cost decrease in the fractional part, which clearly yields the required approximation guarantee.

Each iteration of the rounding scheme involves two steps.
The first step is, for each $t \in [T]$, to independently sample the paths according to the weights $w_t(P)$, upscaled by a factor $K\log \log T$; $K$ is a fixed constant we will choose later.
These paths are added (one by one) to $\mathcal{T}_t$.
The total cost increase due to this step is $O(C \log \log T)$, where $C = \sum_{t \in [T]} \sum_{P \in \mathcal{P}} w_t(P)c(P)$ is the total cost of the fractional part.
Our goal will now be to adjust the fractional paths in a way that reduces the total fractional cost by a constant factor, while maintaining feasibility.
This will clearly lead to our desired approximation ratio: 
each iteration we pay $O(\log \log T)$ times the decrease in the total fractional cost, and once the total fractional cost reaches zero, we have an integral solution.

The main operation that the algorithm will perform in order to achieve this is a ``split and shift'' operation.
Let $(\mathcal{T}, w)$ be the \fps{} solution after the above sampling step.
Let $P$ be some path in $\supp(w_t)$. 
Our goal will be to 
\begin{itemize}
    \item \textbf{(split)} remove some edges from $P$, resulting in a collection of subpaths $P_1, P_2, \ldots, P_q$, which may \emph{not} have their heads in $\mathcal{T}_t$; and then
    \item \textbf{(shift)} for each one of these paths $P_j$, assign its weight to some day $t_j$, in such a way that now $\head(P_j)$ is in $\mathcal{T}_{t_j}$, and $t_j$ is still in the time windows of all the nodes in $P_j$.
\end{itemize}
This would ensure feasibility, while reducing the fractional cost by $w(P)$ times the total cost of the removed edges.
If each edge is removed with constant probability, we obtain the required cost decrease.

In order to make this work, we need some control of the interaction of the different time windows of the nodes on a given path.
The most important fact, immediate from the left aligned structure, is the following.

\begin{lemma}\label{lem:leveltw} 
    If the time windows of $v$ and $w$ overlap, with $\ell(w) \geq \ell(v)$, then
$\twl_w \leq \twl_v$.
\end{lemma}

This means that if we consider a path $P'$ with head $v'$ (which might be a subpath of a path in $\supp(w_t)$ say) where $\ell(v')$ is minimal amongst all nodes in $P'$, then
any $t' \in F_{v'}$ with $t' \leq t$ will be in $F_v$ for all $v \in V(P')$.

The following definition will aid us in shifting always to earlier days,
so that the above can be applied. 

\begin{definition}
    Given a \fps{} $(\mathcal{T}, w)$, for each $v\in V \setminus \{r\}$, 
    let $m_v$ be maximal such that %
    \ifipco{%
                \[ \ssum_{t=m_v}^{\twr_v} \ssum_{P \in \mathcal{P}: P \odot_t v} w_t(P) \geq \frac12. \]
                }{%
                \[ \sum_{t=m_v}^{\twr_v} \sum_{\substack{P \in \mathcal{P}:\\P \odot_t v}} w_t(P) \geq \frac12. \] %
            }
  Then for any $v \in V$, we call $[\twl_v, m_v]$ the \emph{sow} phase of $v$ and $[m_v, \twr_v]$ the \emph{reap} phase of $v$.
  \emph{(Note that the sow and reap phases both contain $m_v$.)}
\end{definition}

Let $S_v$ and $R_v$ denote the sow and reap phases of $v \in V$, at the start of this iteration.
The idea is that we will first 
adjust each path in day $t$ so that, other than the head, $t \in R_v$ for each $v$ on the path.
We do this by simply shortcutting past other nodes, which cannot increase the cost.
This of course reduces the connectivity of the nodes removed from the path.
To fix this, we simply double all weights; 
since the total connectivity of a node attributable to paths not in the reap phase of that node is at most $\tfrac12$, 
this suffices to regain feasibility.
Let $w'$ denote the resulting weights.

At this point, for each path $P$ in $\supp(w'_t)$, $t \in R_v$ for each $v \in P$.
This means that if $P$ were to be shifted to any time in $S_v$ for some $v \in P$, then this is certainly a shift to an earlier time.
We now split $P$ by removing all \emph{fully redundant} edges from it, as per the following definition.

\begin{definition}
    We say that node $v \in V$ has \emph{germinated} if it lies in $V(\mathcal{T}_t)$ for some $t\in S_v$.
    For $i=0, 1,2,\ldots, \log T$ and a path $P$, we say an edge $e\in E(P)$ is $i$-\emph{redundant} if the last node $v$ on $\pprefix{P}{e}$ with level $\ell(v) \leq i$ has germinated, or if all nodes in $\pprefix{P}{e}$ have level strictly larger than $i$. 
    An edge $e \in E(P)$ is \emph{fully redundant} if it is redundant for all $i=0,1,\ldots, \log T $. 
\end{definition}
Note that if $vz \in E(P)$ is fully redundant, then $v$ certainly must have germinated, by considering $i=\ell(v)$.
An important property is the following\ifipco{, proved in the appendix}{}.
\begin{lemma}\label{lem:redsplit}
    Let $Q$ be the set of fully redundant edges on a path $P \in \supp(w'_t)$.
    Then for any path $P' \in P \disconn Q$ with $v := \head(P') \neq \head(P)$, $v$ has the largest level amongst all nodes of $P'$.
\end{lemma}
\ipcoprooflater{redsplit}{lem:redsplit}{}{
    We prove the following, from which the lemma is an immediate consequence.
\begin{inclaim}
  Given a fully redundant edge $e\in E(P)$ and $i \in \znn$, let $u$ be the last node on $\pprefix{P}{e}$ with $\ell(v) \leq i$. 
    Then the edge $e' \in E(P)$ leaving $u$ is also fully redundant.
\end{inclaim}
    Suppose not. Then there is some $j<i$ such that the last node $z\in \pprefix{P}{e'}$ with $\ell(z) \leq j$ has not germinated. 
    But since $u$ is the last node on $\pprefix{P}{e}$ with $\ell(u)\leq i$, clearly $z$ is the last node on $\pprefix{P}{e}$ 
    with $\ell(z) \leq j$. 
    This implies that $e$ is not fully redundant, a contradiction. 
}

%
The final step of the iteration proceeds as follows.
Consider each $t \in [T]$ and $P \in \supp(w_t)$ in turn.
Let $P_1, \ldots, P_q$ be the collection of paths obtained by removing all fully redundant edges from $P$.
Then for each $P_j$, shift it to the latest day $t_j$ for which $t_j \in F_{\head(P_j)}$. 
\begin{lemma}
    The solution is still feasible after the final phase of an iteration.
\end{lemma}
\begin{proof}
    Consider any $t \in [T]$ and path $P \in \supp(w'_t)$, and any $u \in V(P)$.
    Choose $j$ so that $u \in P_j$; let $v = \head(P_j)$.
    We must show that $P_j$ on day $t_j$ supplies connectivity for $u$, i.e., that $P_j \odot_{t_j} u$.
    If $v = \head(P)$ then $t_j = t$ and there is nothing to prove.
    Otherwise, by \Cref{lem:redsplit}, $\ell(v)$ is minimal amongst the levels on the path $P_j$,
    and so by \Cref{lem:leveltw}, $a_u \leq t_j$.
    Thus $t_j \in F_u$, as required. \ipcoqed
\end{proof}
All that remains for the cost bound is to show that each edge is fully redundant with sufficiently large constant probability (enough to counteract the doubling of the weights).
\ifipco{The proof of the following is given in the appendix.}{}
\begin{lemma}\label{lem:costdec}
    For each $t \in [T]$, $P \in \supp(w'_t)$, and $e \in P$, $e$ is fully redundant with probability at least $\tfrac34$, assuming $K$ is chosen sufficiently large.
\end{lemma}
\ipcoskip
{
    \begin{proof}
    Consider some $v \in V$. 
    Since $\sum_{t \in S_v} \sum_{P: P \odot_t v} w_t(P) \geq \tfrac12$, standard calculations imply that 
    the probability none of the paths in this sum are sampled is at most
    $\leq e^{-K\log \log T/2} = \epsilon / \log T$, 
    where $\epsilon = \log(-K/2)$.
  
    Now consider an arbitrary path $P \in \supp(w_t)$ for some $t$ and an arbitrary $e \in E(P)$.
    The edge is fully redundant unless it is not $i$-redundant for some $i=0,1,\ldots,\log T$. 
  Now, $e$ is not $i$-redundant if the last node $u$ on $\pprefix{P}{e}$ with $\ell(u)\leq i$ has not germinated; 
  by the above, this happens with probability at most $\epsilon/\log T$. 
  By a union bound, $e$ is not fully redundant with probability at most 
  \[ (\log T + 1 ) \epsilon /\log T \leq 2 \epsilon.\]
  By choosing $K$ large enough (and hence $\epsilon$ small enough) this is less than $\tfrac14$.
\end{proof}
}
\ifipco{
To complete the proof of \Cref{thm:irpapx}, we should observe that the number of iterations is polynomial (in expectation).
This is fairly clear, and we omit the details in this extended abstract.
}{
To complete the proof of \Cref{thm:irpapx}, we should observe that the algorithm can be implemented to run in polynomial time.
  As a consequence of Lemma~\ref{lem:costdec}, each iteration any unconnected node gets connected to the root with constant probability. 
  It follows that all nodes are connected after $O(\log N)$ iterations, in expectation and with high probability.
  Implementing each iteration in time polynomial in $|\supp(y)|$ and $N$ is straightforward:
  we just need to make sure that the path splitting operations does not cause the number of paths in the support of the \fps{} to explode.  This is not a problem though: since every path must contain at
least one node, every path in the input \fps{} can `generate' at most $N$ new paths down the line, a polynomial increase.
}

\section{Submodular cover over time}\label{sec:sjrp}
Here we consider the \saddc{} problem where in addition to the previously required properties, $f$ is submodular\ifipco{}{(in other words, $f$ is a polymatroid set function)}.
We assume throughout that we have a nice instance, and use $F_v$ to denote the single time window for $v \in V$.

\ifipco{%
First, we observe that the LP relaxation \eqref{lp:saddc} has an equivalent convex formulation in terms of the \lovaszextension{} $\hat{f}$ of $f$. }%
{Recall that the Lovasz-extension of a submodular function $f : 2^V \to \mathbb{R}$ is the function $\hat{f} : [0,1]^V \to \mathbb{R}$ defined as:
\[ \int_0^1 f(\{v : x_v \geq \theta\}) \,d\theta .\]
Note that $\hat{f}$ corresponds exactly to $f$ for integral inputs. 
It is well-known that the \lovaszextension{} of a submodular function is convex~\cite[Sec. 6.3]{fujishige2005submodular}. 
Hence the following program is a convex relaxation of the \subc{} problem. }
\begin{equation}
  \label{lp:lplovasz}
  \begin{aligned}
  \text{min } \qquad &  \ssum_{t\in [T]} \hat{f}(x^t) \\
  \text{s.t.}\qquad & \ssum_{t \in F_v} x_v^t = 1  \qquad \qquad \forall v \in V\\
  & x \geq 0
  \end{aligned}
\end{equation}
\ipcoskip{
In fact, it is precisely equivalent to \eqref{lp:saddc}.
Since $f$ is submodular, there is an optimal solution to \eqref{lp:saddc} such that the support of every day is a chain.
Then it is easy to verify that $\sum_{S} f(S) y^S$ is exactly equal to $\hat{f}(x)$ where $x_v = \sum_{S : v\in S} y^S$ . 
}

\ipcoskip{\subsection{Peliminaries}}
For $x \in [0,1]^V$ and $\theta \in [0,1]$, we define the \emph{level set}
\ifipco{$L_{\theta}(x) = \{ v \in V : x_v\geq \theta \}$.}%
{\[ L_{\theta}(x) = \{ v \in V : x_v\geq \theta \} .\]}
Then
\ifipco{$\hat{f}(x) = \E_\theta[f(L_{\theta}(x))]$, where $\theta \sim \text{Uniform}(0,1)$.}%
{ \[ \hat{f}(x) = \E_\theta[f(L_{\theta}(x))],  \qquad \theta \sim \text{Uniform}(0,1). \] }
\ipcoskip{

This perspective has a nice graphical interpretation which will prove useful in understanding the proof of an important technical lemma later on, and the algorithm in general. If we plot $f(L_{\theta}(x))$, for $\theta$ in the interval $[0,1]$, then $\hat{f}(x)$ is the area between the graph and the horizontal and vertical axes. See Figure~\ref{fig:lovaszplot} for an example.

\begin{figure}[htb!]
  \centering
\def\svgwidth{0.65\textwidth}
\begingroup%
  \makeatletter%
  \providecommand\color[2][]{%
    \errmessage{(Inkscape) Color is used for the text in Inkscape, but the package 'color.sty' is not loaded}%
    \renewcommand\color[2][]{}%
  }%
  \providecommand\transparent[1]{%
    \errmessage{(Inkscape) Transparency is used (non-zero) for the text in Inkscape, but the package 'transparent.sty' is not loaded}%
    \renewcommand\transparent[1]{}%
  }%
  \providecommand\rotatebox[2]{#2}%
  \newcommand*\fsize{\dimexpr\f@size pt\relax}%
  \newcommand*\lineheight[1]{\fontsize{\fsize}{#1\fsize}\selectfont}%
  \ifx\svgwidth\undefined%
    \setlength{\unitlength}{190.82018369bp}%
    \ifx\svgscale\undefined%
      \relax%
    \else%
      \setlength{\unitlength}{\unitlength * \real{\svgscale}}%
    \fi%
  \else%
    \setlength{\unitlength}{\svgwidth}%
  \fi%
  \global\let\svgwidth\undefined%
  \global\let\svgscale\undefined%
  \makeatother%
  \begin{picture}(1,0.43677901)%
    \lineheight{1}%
    \setlength\tabcolsep{0pt}%
    \put(0,0){\includegraphics[width=\unitlength,page=1]{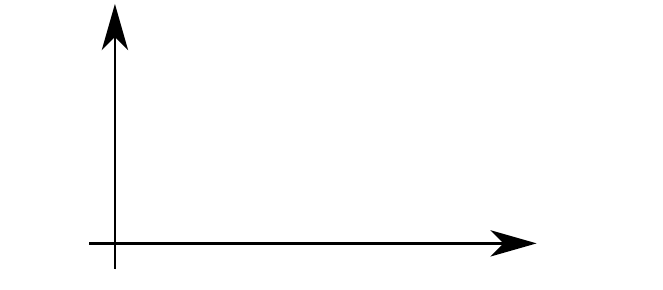}}%
    \put(0.13016522,0.03366856){\color[rgb]{0,0,0}\makebox(0,0)[lt]{\lineheight{1.25}\smash{\begin{tabular}[t]{l}$0$\end{tabular}}}}%
    \put(0.72612133,0.0269199){\color[rgb]{0,0,0}\makebox(0,0)[lt]{\lineheight{1.25}\smash{\begin{tabular}[t]{l}$1$\end{tabular}}}}%
    \put(0.44373371,0.00881269){\color[rgb]{0,0,0}\makebox(0,0)[lt]{\lineheight{1.25}\smash{\begin{tabular}[t]{l}$\theta$\end{tabular}}}}%
    \put(-0.00204197,0.2704631){\color[rgb]{0,0,0}\makebox(0,0)[lt]{\lineheight{1.25}\smash{\begin{tabular}[t]{l}$f(L_{\theta}(x))$\end{tabular}}}}%
    \put(0,0){\includegraphics[width=\unitlength,page=2]{lovaszextension.pdf}}%
    \put(0.33167746,0.11911979){\color[rgb]{0.81960784,0.14509804,0}\transparent{0.91162807}\makebox(0,0)[lt]{\lineheight{1.25}\smash{\begin{tabular}[t]{l}\textbf{\textbf{$\hat{f}(x)$}}\end{tabular}}}}%
  \end{picture}%
\endgroup%

\caption[Graphical interpretation of the \lovaszextension{}.]{Graphical interpretation of the \lovaszextension{} for a polymatroid set function $f$. We have $\theta$ on the horizontal axis, $f(L_{\theta}(x))$ on the vertical axis and $\hat{f}(x)$ is the area between the graph and axes. }
  \label{fig:lovaszplot}
\end{figure}

}
Define the truncation $\trunc{x}{\theta}$ by $\trunc{x}{\theta}_v = \min\{x_v, \theta \}$.
\ipcoskip{
Observe that
\[ \hat{f}(\trunc{x}{\theta}) = \int_{0}^\theta f(L_{\eta}(x)) \,d\eta.  \] 
}
\ipcoskip{Finally, we introduce the concept of an $\alpha$-supported set in Definition~\ref{def:alphasupported}. For now, think of these as level sets for which setting the corresponding entries to zero would yield a large reduction in the cost of the \lovaszextension{}.}
\begin{definition}
  \label{def:alphasupported}
  Given $\theta \in [0,1]$ and $x\in [0,1]^V$, we say that the set $L_{\theta}(x)$ is \emph{$\alpha$-supported} if:
      \begin{equation}
      \label{eq:alphasupported}
    \hat{f}(x) - \hat{f}(\trunc{x}{\theta}) \geq \alpha f(L_{\theta}(x)).
  \end{equation}
\end{definition}
We will provide some \ipcoskip{more} intuition for this definition later, but first we describe the algorithm.

\medskip

\ipcoskip{\subsection{The algorithm}}
    \begin{algorithmic}[1]
        \Require{A solution $x$ to \eqref{lp:lplovasz}.}
        \Ensure{A solution $(S_t)_{t \in [T]}$ to the submodular cover over time problem.}
    \State $S_t \gets \emptyset$ for all $t\in [T]$.
    \For{$i=1,\dots,\log T$}
    \For{$t \in [T]$} \label{stp:lovsupp}
    \If{there exists $\theta \in [0,1]$ such that $L_\theta(x^t)$ is $\frac{1}{32\log\log T}$-supported}
    \State\label{step:add} Choose such a $\theta$ and set
$S_t \gets S_t \cup  L_{\theta}(x^t)$, $x^t \gets \truncss{x}{t}{\theta}$. 
    \Else
    \State \label{step:leaf} $S_t \gets S_t \cup L_1(x^t)$. 
      \EndIf
      \EndFor

    \State\label{step:merge}
Merge time periods by setting, for all $t \in\{ k 2^i : k=0,1,2,\dots\}$,
\ifipco{$\qquad$ \phantom{ddddd}
      $x^{t + 1} \gets x^{t + 1} + x^{t + 2^{i-1}+ 1}\quad$ and $\quad x^{t + 2^{i-1} + 1} \gets 0$.
      }{
\begin{equation*}
      x^{t + 1} \gets x^{t + 1} + x^{t + 2^{i-1}+ 1} \quad \text{ and } \quad  x^{t + 2^{i-1} + 1} \gets 0 .
  \end{equation*}
  }
\EndFor
\State \Return{$(S_t)_{t\in [T]}$.} 
\end{algorithmic}

\paragraph{Feasibility.}
The only steps where the coverage $\sum_{t \in F_v} x^t_v$ for an item $v$ could possibly decrease are steps
\ref{step:add} and \ref{step:merge}.
In step \ref{step:add}, $v$ is added to $S^t$, so this is clearly no problem.
In step \ref{step:merge}, we are shifting weight from some time $t+2^{i-1}+1$ to $t+1$\ifipco{}{(see \Cref{fig:mergetime} for an illustration of this step)}.
This cannot leave the time window of $v$ unless $\ell(v) < i$.
But if $v$ has not been covered by the end of iteration $\ell(v)$,
all of the fractional coverage for $v$ will have been merged into the left endpoint of its time window $t' = \min F_v$. 
This means that $L_\theta(x^{t'})$ contains $v$ for any $\theta$, ensuring $v$ will be added to
$S_{t'}$ in step~\ref{stp:lovsupp} of iteration $\ell(v)+1$. 

\ipcoskip{
\begin{figure}[htb!]
  \centering
\def\svgwidth{0.65\textwidth}
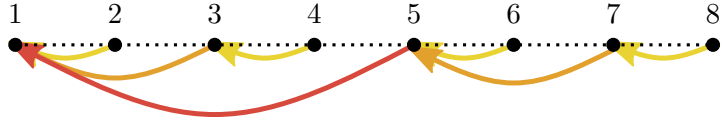
\caption[Illustration of time unit merging in step~\ref{step:merge}.]{Illustration of time unit merging in step~\ref{step:merge} on time horizon $[T] = \{1,\dots,8\}$ over $3$ iterations.  In the first iteration, even time units are merged into the preceding odd time units, in the second iteration, $3$ and $7$ are merged into $1$ and $5$ respectively, and finally $5$ is merged into $1$.}
  \label{fig:mergetime}
\end{figure}
}

\paragraph{Cost analysis.}
This is where the key insights lie.
The main driver is the following technical lemma\ifipco{(the proof is postponed to the appendix)}{}.
\begin{lemma}
  \label{lem:concentration}
  For any $x \in [0,1]^V$ and $\alpha\in (0,1]$, 
\ifipco{either there exists $\theta \in [0,1]$ such that $L_{\theta}(x)$ is $\frac{\alpha}{32}$-supported, or otherwise
    $2^{1/\alpha}f(L_{1}(x))   \leq\hat{f}(x)$.
}{
  at least one of the following holds:
  \begin{enumerate}
    \item\label{case:con1} there exists $\theta \in [0,1]$ such that $L_{\theta}(x)$ is $\frac{\alpha}{32}$-supported.
    \item\label{case:con2} $2^{1/\alpha}f(L_{1}(x))   \leq\hat{f}(x)$.
\end{enumerate}
}
\end{lemma}
\ifipco{
The intuition for this lemma is that if \ifipco{no $\theta$ with the desired property exists}{(\ref{case:con1}) fails to hold}, it can be shown that $f(L_{\theta}(x))$ decreases quickly everywhere, 
and consequently $f(L_1(x))$ is small compared to $\hat{f}(x) = \E_\theta [f(L_\theta(x))]$.
}{
Before we formally prove this lemma, let us generate some graphical intuition. 
Recall that $L_{\theta}(x)$ is $\alpha$-supported if 
\[ {\alpha} f(L_\theta(x)) \leq  \hat{f}(x) - \hat{f}(\trunc{x}{\theta}).\]
In Figure~\ref{fig:concentration}, the plots of the \lovaszextension{} on two different inputs are shown.  
On the left we have Case 1 of the lemma. 
If we pick the set $L_{\theta'}(x)$ as indicated, the cost is the area of the light blue square. 
However, the dark blue area containing $\hat{f}(x) - \hat{f}(\trunc{x}{\theta'})$ covers a large fraction of that cost. 
On the right, no such set exists. 
This implies the graph must
decrease quite quickly everywhere and therefore lie entirely above some exponential curve. 
Consequently the cost of the set $L_1(x)$, which is the light green rectangle, is small compared to $\hat{f}(x)$, which is the area between the graph and the axes.   

\begin{figure}[htb!]
  \centering
\def\svgwidth{0.7\textwidth}
\begingroup%
  \makeatletter%
  \providecommand\color[2][]{%
    \errmessage{(Inkscape) Color is used for the text in Inkscape, but the package 'color.sty' is not loaded}%
    \renewcommand\color[2][]{}%
  }%
  \providecommand\transparent[1]{%
    \errmessage{(Inkscape) Transparency is used (non-zero) for the text in Inkscape, but the package 'transparent.sty' is not loaded}%
    \renewcommand\transparent[1]{}%
  }%
  \providecommand\rotatebox[2]{#2}%
  \newcommand*\fsize{\dimexpr\f@size pt\relax}%
  \newcommand*\lineheight[1]{\fontsize{\fsize}{#1\fsize}\selectfont}%
  \ifx\svgwidth\undefined%
    \setlength{\unitlength}{266.71447369bp}%
    \ifx\svgscale\undefined%
      \relax%
    \else%
      \setlength{\unitlength}{\unitlength * \real{\svgscale}}%
    \fi%
  \else%
    \setlength{\unitlength}{\svgwidth}%
  \fi%
  \global\let\svgwidth\undefined%
  \global\let\svgscale\undefined%
  \makeatother%
  \begin{picture}(1,0.5484089)%
    \lineheight{1}%
    \setlength\tabcolsep{0pt}%
    \put(0,0){\includegraphics[width=\unitlength,page=1]{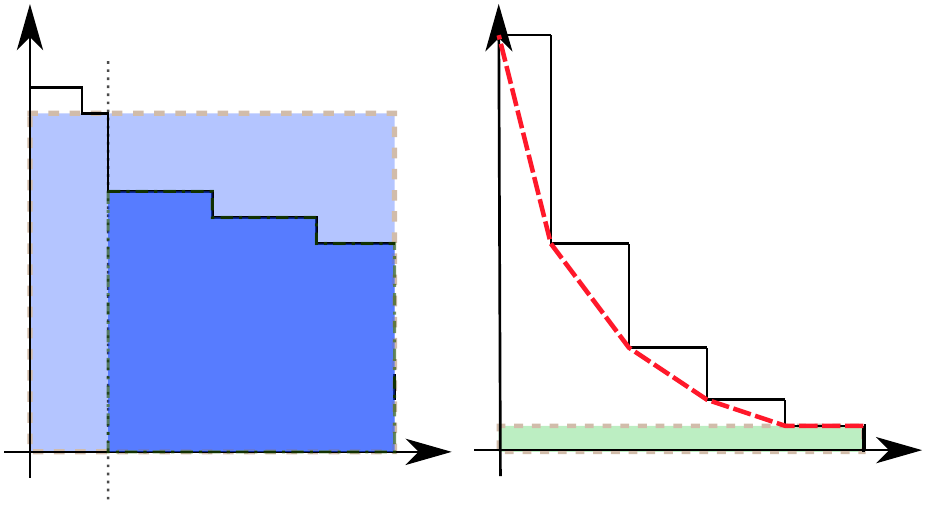}}%
    \put(0.08237466,0.01125151){\color[rgb]{0,0,0}\makebox(0,0)[lt]{\lineheight{1.25}\smash{\begin{tabular}[t]{l}$\theta'$\end{tabular}}}}%
    \put(0.20316349,0.37337456){\color[rgb]{0,0,0}\makebox(0,0)[lt]{\lineheight{1.25}\smash{\begin{tabular}[t]{l}$f(L_{\theta'}(x))$\end{tabular}}}}%
    \put(0.1745705,0.1613008){\color[rgb]{0,0,0}\makebox(0,0)[lt]{\lineheight{1.25}\smash{\begin{tabular}[t]{l}$\hat{f}(x) - \hat{f}(\trunc{x}{\theta'})$\end{tabular}}}}%
    \put(-1.32163814,-0.1405998){\color[rgb]{0,0,0}\makebox(0,0)[lt]{\begin{minipage}{1.40599791\unitlength}\raggedright \end{minipage}}}%
    \put(-1.51847786,-0){\color[rgb]{0,0,0}\makebox(0,0)[lt]{\begin{minipage}{1.99651696\unitlength}\raggedright \end{minipage}}}%
  \end{picture}%
\endgroup%

\caption[Illustration of Lemma~\ref{lem:concentration}.]{Illustration of Lemma~\ref{lem:concentration}. Either there is a set $L_{\theta'}$ whose cost gets covered for a large fraction by $\hat{f}(x) - \hat{f}(\trunc{x}{\theta'})$, or, the plot of $f(L_{\theta}(x))$ lies above some exponential curve and therefore $f(L_1(x))$ is small compared to $\hat{f}(x)$.}
  \label{fig:concentration}
\end{figure}
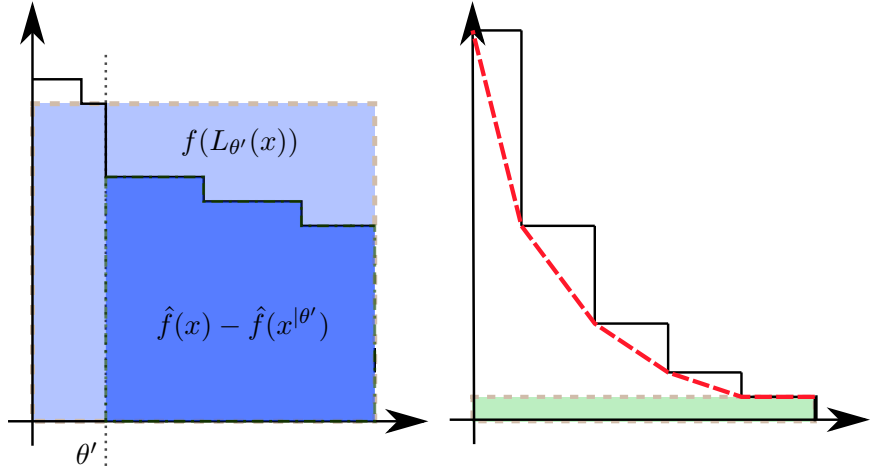
}

\ipcoproofoflater{concentration}{lem:concentration}{}{
 Let $k\in \mathbb{N}$ be such that $\frac1{16} \alpha \leq \frac1k \leq \frac18\alpha$. Note that this implies $k\geq 8$. 
 Let $z = f(L_1(x))$. 
 
 \begin{inclaim}If $2^{1/\alpha}z > \hat{f}(x)$, there exists $m \in \mathbb{N}$ such that 
   \begin{equation} \label{eq:conclaim} f(L_{\frac{k-m}{k}}(x)) < 2^m z.\end{equation}
   \end{inclaim}
   Before we prove the claim, let's see that it implies the lemma. Suppose that $2^{1/\alpha} f(L_1(x)) > \hat{f}(x)$, since otherwise we are done. The condition of the claim then holds, so take the smallest $m$ that satisfies~\eqref{eq:conclaim}, and let $\theta = \frac{k-m}{k}$. We claim that
   \begin{align*}
       \frac{\alpha}{32} f(L_\theta(x)) \leq \hat{f}(x) - \hat{f}(\trunc{x}{\theta}) .
    \end{align*}
    To see this we first rewrite the right hand side as an integral. 
    \begin{align} 
      \label{eq:concentration1}
      \hat{f}(x) - \hat{f}(\trunc{x}{\theta}) = \int_0^1 f(L_\eta(x)) \,d\eta - \int_0^1 f(L_\eta(\trunc{x}{\theta})) \,d\eta = \\\nonumber
     = \int_0^1 f(L_\eta(x)) \,d\eta - \int_0^\theta f(L_\eta(x)) \,d\eta = \int_\theta^1 f(L_\eta(x)) \,d\eta .
 \end{align}
 Recall that $f(L_\eta(x))$ is monotonically decreasing and that $m\geq 1$ so that $\theta+\frac1k = \frac{k-m+1}{k} \leq 1$. Then
    \begin{align} 
      \label{eq:concentration2}
      \int_\theta^1 f(L_\eta(x)) \,d\eta \geq  \int_\theta^{\theta+ \frac1k} f(L_\eta(x)) \,d\eta \geq \frac1k f(L_{\theta +\frac1k }(x)). 
 \end{align}
 Finally, we use the fact that $m$ is minimal, which implies that $f(L_{\frac{k-m+1}{k}}(x)) \geq 2^{m-1}z$, together with~\eqref{eq:concentration1} and~\eqref{eq:concentration2}: 
    \begin{align} 
      \label{eq:concentration3}
      \hat{f}(x) - \hat{f}(\trunc{x}{\theta}) \geq  \tfrac1k 2^{m-1}z = \tfrac1{2k} 2^m z > \tfrac{\alpha}{32} f(L_{\theta}(x)).
 \end{align}
 In the final inequality of~\eqref{eq:concentration3} we use that the fact that we chose $m$ to satisfy $2^m z > f(L_{\frac{k-m}{k}}(x))$ and $\frac1k \geq \frac1{16} \alpha$.

 Now we proceed to prove the claim. Suppose for contradiction that the condition of the claim holds but no $m$ satisfies inequality~\eqref{eq:conclaim}. Then, in particular it must hold that $f(L_{\frac1k}(x)) \geq 2^{k-1}z$ and therefore we obtain
 \[
     \hat{f}(x) \geq \int_0^{\frac1k} f(L_\eta(x)) \,d\eta \geq \frac1k f(L_{\frac1k}(x)) \geq \frac1k 2^{k-1} z.
 \]
 Since $k \geq 8$, $\tfrac1k 2^{k-1}z \geq 2^{k/2}z$.
 Since also $\frac1k \leq \frac18 \alpha$, we deduce
    \[ \hat{f}(x) \geq 2^{k/2}z \geq 2^{4/\alpha} z \geq 2^{1/\alpha} z,
    \]
    contradicting that $2^{1/\alpha}z >  \hat{f}(x)$. This proves the claim, and hence the lemma.
}

Consider now some iteration $i$ of the algorithm, and a particular choice of $t$.
If step \ref{step:add} is executed, then $\hat{f}(x^t)$ decreases by at least a $\frac{1}{32\log\log T}$ fraction of $f(L_\theta(x^t))$, which is an upper bound on the cost increase of the current solution by subadditivity.
Otherwise, by \Cref{lem:concentration} (with $\alpha=\frac{1}{\log\log T}$), 
$f(L_1(x^t)) \leq \hat{f}(x^t)/\log T$.
The total cost of sets chosen in step~\ref{step:leaf} in a single iteration is thus at most $\sum_{t \in [T]} \hat{f}(x^t)/\log T$.
So over all $\log T$ iterations, this incurs a total cost not more than the original objective value of the relaxation.

\ifipco{
Again, to complete the proof of \Cref{thm:sjrpapx}, we should argue that the algorithm runs in polynomial time; we postpone the straightforward details to the full version.
}{
  It remains to argue that we can implement the procedure in polynomial time. 
  There are $\log T$ iterations, each of which requires looping over at most $T$ days.
  The only difficult step is finding an 
  $\frac{1}{32\log \log T}$-supported set, if it exists. 
  But there are at most $N$ distinct values of $\theta$ at which $L_{\theta}(x^t)$ changes, and we only need to check these. 
}

\bibliographystyle{abbrv}
\bibliography{lit}

\end{document}